\documentclass[11pt]{article}
\usepackage{enumerate}
\usepackage{amsfonts}
\usepackage{amssymb}
\usepackage{amsmath}
\usepackage{xspace}
\usepackage{algpseudocode}

\usepackage{graphicx}
\usepackage{multicol}
\usepackage[margin=1in]{geometry}
\usepackage{hyperref}
\usepackage{amsthm}

\usepackage{times}

\newtheorem{lemma}{Lemma}
\newtheorem{theorem}[lemma]{Theorem}

\newtheorem{xmpl}[lemma]{Example}
\newenvironment{example}{\begin{xmpl}\rm}{\end{xmpl}}

\newtheorem{rmark}[lemma]{Remark}
\newenvironment{remark}{\begin{rmark}\rm}{\end{rmark}}

\theoremstyle{definition}
\newtheorem{definition}[lemma]{Definition}

\usepackage[T1]{fontenc}

\usepackage{amssymb,amsfonts,amsthm}
\usepackage{changepage}
\usepackage{chngcntr}
\usepackage[inline,shortlabels]{enumitem}
\usepackage{mathtools}
\usepackage{scrextend}
\usepackage{stmaryrd }
\usepackage{thmtools}
\usepackage{thm-restate}
\usepackage{tikz}
\usetikzlibrary{chains,matrix,scopes,decorations.shapes,arrows,shapes}
\usepackage{wrapfig}
\usepackage{xspace}

\numberwithin{equation}{section}

\makeatletter
\let\epsilon\varepsilon
\let\phi\varphi

\makeatother

\newcommand{\MP}{\textsf{MP}}

\newcommand{\RT}{\textsf{RT}}

\newcommand{\lolli}{\G_{\bowtie}}

\newcommand{\G}{{\cal G}}

\renewcommand{\S}{{\cal S}}

\newcommand{\PO}{Player~$1$\xspace}
\newcommand{\PT}{Player~$2$\xspace}
\newcommand{\PLi}{Player~$i$\xspace}

\newcommand{\thresh}{\texttt{Th}\xspace}
\newcommand{\Max}{\text{Max}\xspace}
\newcommand{\Min}{\text{Min}\xspace}

\newcommand{\Real}{\mathbb{R}}
\newcommand{\Nat}{\mathbb{N}}
\newcommand{\Rat}{\mathbb{Q}}

\newcommand{\zug}[1]{\langle #1  \rangle}
\newcommand{\set}[1]{\{ #1 \}}
\newcommand{\stam}[1]{}

\newcommand{\val}{\textsf{val}}
\newcommand{\play}{\textsf{play}}
\newcommand{\supp}{\text{supp}}
\newcommand{\payoff}{\textsf{payoff}}

\newcommand{\Genpot}[2]{\textsf{Pot}(#1,#2)}

\newcommand{\Admi}{\textsc{Adm}}

\title{\Large Bidding Graph Games with Partially-Observable Budgets\thanks{This research was supported in part by ISF grant no. 1679/21, by the ERC CoG 863818 (ForM-SMArt), and the European Union’s Horizon 2020 research and innovation programme under the Marie Skłodowska-Curie Grant Agreement No. 665385.}}
\author{Guy Avni\thanks{University of Haifa}
\and Ism\"ael Jecker\thanks{University of Warsaw} \and {\DJ}or{\dj}e \v{Z}ikeli\'c\thanks{Institute of Science and Technology Austria (ISTA)}}

\date{}
\begin{document}
\maketitle

\begin{abstract}
Two-player zero-sum {\em graph games} are a central model, which proceeds as follows. 
A token is placed on a vertex of a graph, and the two players move it to produce an infinite {\em play}, which determines the winner or payoff of the game. Traditionally, the players alternate turns in moving the token. In {\em bidding games}, however, the players have budgets and in each turn, an auction (bidding) determines which player moves the token. So far, bidding games have only been studied as full-information games. 
In this work we initiate the study of partial-information bidding games: we study bidding games in which a player's initial budget is drawn from a known probability distribution. 
We show that while for some bidding mechanisms and objectives, it is straightforward to adapt the results from the full-information setting to the partial-information setting, for others, the analysis is significantly more challenging, requires new techniques, and gives rise to interesting results. 
Specifically, we study games with {\em mean-payoff} objectives in combination with {\em poorman} bidding. We construct optimal strategies for a partially-informed player who plays against a fully-informed adversary. We show that, somewhat surprisingly, the {\em value} under pure strategies does not necessarily exist in such games. 
\end{abstract}

\section{Introduction}\label{sec:intro}
We consider two-player zero-sum {\em graph games}; a fundamental model with applications, e.g., in multi-agent systems \cite{AHK02}. 
A graph game is played on a finite directed graph as follows. A token is placed on a vertex and the players move it throughout the graph to produce an infinite path, which determines the payoff of the game. 
Traditional graph games are {\em turn-based}: the players alternate turns in moving the token. 

{\em Bidding games} \cite{LLPU96,LLPSU99} are graph games in which an ``auction'' (bidding) determines which player moves the token in each turn. The concrete bidding mechanisms that we consider proceed as follows. In each turn, both players simultaneously submit bids, where a bid is legal if it does not exceed the available budget. The higher bidder ``wins'' the bidding and moves the token. The mechanisms differ in their payment schemes, which are classified according to two orthogonal properties. {\em Who pays:} in {\em first-price} bidding only the higher bidder pays the bid and in {\em all-pay} bidding both players pay their bids. {\em Who is the recipient:} in {\em Richman} bidding (named after David Richman) payments are made to the other player and in {\em poorman} bidding payments are made to the ``bank'', i.e., the bid is lost. As a rule of thumb, bidding games under all-pay and poorman bidding are respectively technically more challenging than first-price and Richman bidding. More on this later. In terms of applications, however, we argue below that poorman bidding is often the more appropriate bidding mechanism.

\paragraph{Applications.}
A central application of graph games is {\em reactive synthesis}~\cite{PR89}: given a specification, the goal is to construct a {\em controller} that ensures correct behavior in an adversarial environment. Synthesis is solved by constructing a turn-based graph game in which \PO is associated with the controller and \PT with the environment, and searching for a winning \PO strategy. 


Bidding games extend the modeling capabilities of graph games. For example, they model ongoing and stateful auctions in which budgets do not contribute to the players' utilities. 
Advertising campaigns are one such setting: the goal is to maximize visibility using a pre-allocated advertising budget. 
By modeling this setting as a bidding game and solving for \PO, we obtain a bidding strategy with guarantees against any opponent\footnote{A worst-case modelling assumes that the other bidders cooperate against \PO.}. Maximizing visibility can be expressed as a mean-payoff objective (defined below).

All-pay poorman bidding is particularly appealing since it constitutes a dynamic version of the well-known Colonel Blotto games \cite{Bor21}. Rather than thinking of the budgets as money, we think of them as resources at the disposal of the players, like time or energy. Then, deciding how much to bid represents the effort that a player invests in a competition, e.g., investing time to prepare for a job interview, where the player that invests more wins the competition. 


\paragraph{Prior work -- full-information bidding games.}
The central quantity in bidding games is the {\em initial ratio} between the players' budgets. Formally,  for $i \in \set{1,2}$, let $B_i$ be \PLi's initial budget. Then, \PO's initial ratio is $B_1/(B_1 + B_2)$. 
A {\em random-turn game} \cite{PSSW09} with parameter $p \in [0,1]$ is similar to a bidding game only that instead of bidding, in each turn, we toss a coin with probability $p$ that determines which player moves the token. 
Formally, a random-turn game is a special case of a {\em stochastic game} \cite{Con92}.

{\em Qualitative objectives.} 
In {\em reachability} games, each player is associated with a target vertex, the game ends once a target is reached, 
and the winner is the player whose target is reached.
Reachability bidding games were studied in \cite{LLPU96,LLPSU99}. It was shown that, for first-price reachability games, a {\em threshold ratio} exists, which, informally, is a necessary and sufficient initial ratio for winning the game. 
Moreover, it was shown that first-price Richman-bidding games are equivalent to uniform random-turn games (and only Richman bidding); namely, the threshold ratio in a bidding game corresponds to the value of a uniform random-turn game. 
All-pay reachability games are technically more challenging. Optimal strategies might be mixed and may require sampling from infinite-support probability distributions even in extremely simple games \cite{AIT20}.


{\em Mean-payoff games.} 
Mean-payoff games are infinite-duration quantitative games. Technically, each vertex of the graph is assigned a weight, and the {\em payoff} of an infinite play is the long-run average sum of weights along the path. 
The payoff is \PO's reward and \PT's cost, thus we refer to them respectively as  \Max and \Min. For example, consider the ``bowtie'' game $\lolli$, depicted in Fig.~\ref{fig:bowtie}.  The payoff in $\lolli$ corresponds to the ratio of bidding that \Max wins. Informally $\lolli$ models the setting in which in each day a publisher sells an ad slot, and \Max's objective is to maximize visibility: the number of days that his ad is displayed throughout the year. Unlike reachability games, intricate equivalences between mean-payoff bidding games and random-turn games are known for all the mechanisms described above~\cite{AHC19,AHI18,AHZ21,AJZ21}.


\begin{example}
We illustrate the equivalences between full-information bidding games and random-turn games. 
Consider the ``bowtie'' game $\lolli$ (see Fig.~\ref{fig:bowtie}). For $p \in [0,1]$, the random-turn game $\RT(\lolli, p)$ that uses a coin with bias $p$ is depicted in Fig.~\ref{fig:RTBowtie}. Its expected payoff is $p$. 

Suppose that the initial ratio is $r \in (0,1)$. Under first-price Richman-bidding, the optimal payoff in $\lolli$ does not depend on the initial ratio: no matter what $r$ is, the optimal payoff that \Max can guarantee is arbitrarily close to $0.5$, hence the equivalance with $\RT(\lolli, 0.5)$. Under first-price poorman bidding, the optimal payoff {\em does} depend on the initial ratio: roughly, the optimal payoff that \Max can guarantee is $r$, hence the equivalence with $\RT(\lolli, r)$.
For all-pay bidding, pure strategies are only ``useful'' in all-pay poorman bidding and only when $r > 0.5$, where \Max can guarantee an optimal payoff of $\frac{2r -1}{r}$.  
The results extend to general strongly-connected games (see Thm.~\ref{thm:FI-MP}).\hfill$\triangle$
\end{example}

 \begin{figure}
    \begin{minipage}[b]{0.45\linewidth}
    \centering
\includegraphics[height=1.3cm]{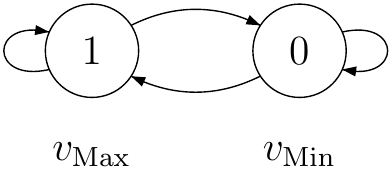}
\caption{The mean-payoff game $\lolli$ with the weights in the vertices.}
\label{fig:bowtie}
\end{minipage}
\hspace{0.05\linewidth}
  \begin{minipage}[b]{0.45\linewidth}
    \centering
\includegraphics[height=1.3cm]{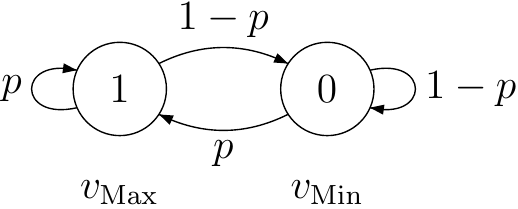}
\caption{The simplified random-turn game $\RT(\lolli, p)$, for $p \in [0,1]$.}
\label{fig:RTBowtie}
\end{minipage}
\end{figure}

\paragraph{Our contributions -- partial-information bidding games.}
In most auction domains, bidders are not precisely informed of their opponent's budget. Bidding games, however, have only been studied as full-information games. 
We initiate the study of bidding games in which the players are partially informed of the opponent's budget. Specifically, we study bidding games in which the two players' budgets are drawn from a known probability distribution, and the players' goal is to maximize their expected utility. 
We first show that the results on qualitative objectives as well as first-price Richman bidding transfer to the partial-information setting.

 
We turn to study mean-payoff poorman-bidding games, which are significantly more challenging. We focus on {\em one-sided} partial-information games in which only \PT's budget is drawn from a probability distribution. 
Thus, \PO is {\em partially informed} and \PT is {\em fully informed} of the opponent's budget.
We argue that one-sided partial-information games are practically well-motivated. Indeed, one-sided partial information is a  worst-case modelling: the utility that an optimal strategy for \PO guarantees in the game, is a lower bound on the utility that it will guarantee when deployed against the concrete environment. 
We illustrate our results in the following example. 

\begin{example}
\label{ex:partially-informed}
Consider the bowtie game $\lolli$ (Fig.~\ref{fig:bowtie}), where \Max (the partially-informed player) starts with a budget of $B$ and \Min (the fully-informed player) starts with a budget that is drawn uniformly at random from $\supp(\gamma) = \set{C_1,C_2}$. We describe an optimal strategy for \Max under first-price poorman bidding. \Max carefully chooses an $x \in [B \cdot \frac{C_1}{C_2},B]$ and divides his budget into two ``wallets''; the first with budget $x$ and the second with budget $B-x$. He initially uses his first wallet to play an optimal full-information strategy assuming the initial budgets are $x$ and $C_1$, which guarantees a payoff of at least $p_1 = \frac{x}{C_1 + x}$. If \PT spends more than $C_1$, i.e., her initial budget was in fact $C_2$, then \PO proceeds to use his second wallet against \PT's remaining budget, which guarantees a payoff of at least $p_2 = \frac{B-x}{B-x+C_2-C_1}$. Thus, the expected payoff is at least $0.5\cdot (p_1 + p_2)$, and \Max simply chooses an $x$ that maximizes this expression. Note that the constraint that $x \geq B \cdot \frac{C_1}{C_2}$ implies that $p_1 \geq p_2$, 
thus \Min has an incentive to play so that \Max proceeds to use his second wallet. We show that this strategy is optimal, and extend the technique to obtain optimal strategies in general strongly-connected games for first-price and all-pay poorman bidding. 

Finally, we show that the optimal payoff that \Min can guarantee in $\lolli$, is obtained by a surprisingly simple strategy. 
We show that the following \Min strategy is optimal: when her initial budget is $C_i$, for $i \in \set{1,2}$, \Min follows an optimal full-information strategy for ratio $B/(B+C_i)$. That is, she ``reveals'' her true budget in the first round and cannot gain utility by hiding this information. The technical challenge is to show that this strategy is optimal.
\hfill$\triangle$
\end{example}

Our results show that contrary to turn-based, stochastic games, and full-information bidding games, there is a gap between the optimal payoffs that the players can guarantee with pure strategies. Thus, the {\em value} does not necessary exist in partial-information mean-payoff bidding games under pure strategies.

\paragraph{Related work.}
The seminar book \cite{AMS95} studies the mean-payoff game $\lolli$ under one-sided partial-information with a different semantic to the one we study. Let $L$ or $R$ denote the two vertices of $\lolli$. \Min has partial information of the weights of $L$ and $R$, which, before the game begins, are drawn from a known probability distribution. \Max, the fully-informed player, knows the weights. In each turn, \Max chooses $L$ or $R$, followed by \Min who either ``accepts'' or ``rejects'' \Max's choice, thus both players can affect the movement of the token. The value in the game is shown to exist. Interestingly and similar in spirit to our results, there are cases in which \Max cannot use his knowledge advantage and his optimal strategy reveals which of the two vertices he prefers. One-sided partial information have also been considered in turn-based graph games, e.g., \cite{Reif84,RC+07,DDR06}.

{\em Discrete bidding games} were studied in \cite{DP10}; namely, budgets are given in coins, and the minimal positive bid a player can make is a single coin. Tie-breaking is a significant factor in such games~\cite{AAH21}. Non-zero-sum bidding games were studied in \cite{MKT18}. See also the survey~\cite{AH20}.

\section{Preliminaries}
\label{sec:prelim}
\paragraph{Strategies in bidding games.}
A bidding game is played on a directed graph $\zug{V, E}$. A {\em strategy} in any graph game is a function from histories to actions. In bidding games, a history consists of the sequence of vertices that were visited and bids made by the two players. We stress that the history does not contain the current state of the budgets. Rather, a player can compute his opponent's current budget based on the history of bids, if he knows her initial budget. We formalize the {\em available budget} following a history. For $i \in \set{1,2}$, suppose the initial budget of \PLi is $B_i$. For a history $h$, we define the {\em investments} of \PLi throughout $h$, denoted $\text{Inv}_i(h)$. In all-pay bidding, $\text{Inv}_i(h)$ is the sum bids made by \PLi throughout $h$, and in first-price bidding, it is the sum only over the winning bids. We denote by $B_i(h)$ \PLi's  available budget following $h$. Under Richman bidding, winning bids are paid to the opponent, thus $B_i(h) = B_i - \text{Inv}_i(h) + \text{Inv}_{3-i}(h)$. Under poorman bidding, winning bids are paid to the bank, thus $B_i(h) = B_i - \text{Inv}_i(h)$.

Given a history, a strategy prescribes an action, which in a bidding game, is a pair $\zug{b,u} \in \Real \times V$, where $b$ is a bid and $u$ is the vertex to move to upon winning. We restrict the actions of the players following a history $h$ so that (1) the bid does not exceed the available budget, thus following a history $h$, a legal bid for \PLi is a bid in $[0, B_i(h)]$, and (2) a player must choose a neighbor of the vertex that the token is placed on.  We restrict attention to strategies that choose legal actions for all histories. Note that we consider only {\em pure} strategies and disallow {\em mixed} strategies (strategies that allow a random choice of action).

\begin{definition}
For $i \in \set{1,2}$, we denote by $\S_i(B_i)$ the set of legal strategies for \PLi with an initial budget of $B_i$. Note that with a higher initial budget, there are more strategies to choose from, i.e., for $B'_i > B_i$, we have $\S_i(B_i) \subseteq \S_i(B'_i)$.
\end{definition}

The central quantity in bidding games is the initial ratio, defined as follows.
\begin{definition}
{\bf Budget ratio.}
When \PLi's budget is $B_i$, for $i \in \set{1,2}$, we say that \PLi's ratio is $\frac{B_i}{B_1 + B_2}$. 
\end{definition}

\paragraph{Plays.} Consider initial budgets $B_1$ and $B_2$ for the two players, two strategies $f \in \S_1(B_1)$ and $g \in \S_2(B_2)$, and an initial vertex $v$. The triple $f$, $g$, and $v$ gives rise to a unique {\em play}, denoted $\play(v, f,g)$. The construction of $\play(v, f, g)$ is inductive and is intuitively obtained by allowing the players to play according to $f$ and $g$. Initially, we place the token on $v$, thus the first history of the game is $h=v$. Suppose a history $h$ has been played. Then, the next action that the players choose is respectively $\zug{u_1, b_1} = f(h)$ and $\zug{u_2, b_2} = g(h)$. If $b_1 > b_2$, then \PO wins the bidding and the token moves to $u_1$, and otherwise \PT wins the bidding and the token moves to $u_2$. Note that we resolve ties arbitrarily in favor of \PT. The play continues indefinitely. Since the players always choose neighboring vertices, each play corresponds to an infinite path in $\zug{V, E}$. For $n \in \Nat$, we use $\play_n(v, f,g)$ to denote its finite prefix of length $n$. We sometimes omit the initial vertex from the play when it is clear from the context.

\paragraph{Objectives.}
We consider zero-sum games. An {\em objective} assigns a {\em payoff} to a play, which can be thought of as \PO's reward and \PT's penalty. We thus sometimes refer to \PO as \Max and \PT as \Min. We denote by $\payoff(f, g, v)$ the payoff of the play $\play(f, g, v)$. 

\paragraph{Qualitative objectives.} The payoff in games with qualitative objectives is in $\set{-1,1}$. We say that \PO\ {\em wins} the play when the payoff is $1$. We consider two qualitative objectives. (1) {\em Reachability.} There is a distinguished target vertex $t$ and a play is winning for \PO iff it visits $t$. (2) {\em Parity.} Each vertex is labeled by an index in $\set{1,\ldots,d}$ and a play is winning for \PO iff the highest index that is 
\stam{
\begin{itemize}
\item {\bf Reachability.} There is a distinguished target vertex $t$ and a play is winning for \PO iff it visits $t$. 
\item {\bf Parity.} Each vertex is labeled by an index in $\set{1,\ldots,d}$ and a play is winning for \PO iff the highest index that is visited infinitely often is odd. 
\end{itemize}
}
Parity objectives are important in practice, e.g., reactive synthesis~\cite{PR89} is reducted to the problem of solving a (turn-based) parity games. 

\paragraph{Mean-payoff games.} 
The quantitative objective that we consider is {\em mean-payoff}. Every vertex $v$ in a mean-payoff game has a weight $w(v)$ and the payoff of an infinite play is the long-run average weight that it traverses. Formally, the payoff of an infinite path $v_1, v_2, \ldots$ is $\liminf_{n \to \infty} \frac{1}{n} \sum_{1 \leq i < n} w(v_i)$. Note that the definition favors \Min since it uses $\liminf$. 

\paragraph{Values in full-information bidding games.}
We are interested in finding the optimal payoff that a player can {\em guarantee} with respect to an initial budget ratio. 
Let $c \in \Real$ and initial budgets $B_1$ and $B_2$. We say that \PO can {\em guarantee} a payoff of $c$, if he can reveal that he will be playing according to a strategy $f \in \S_1(B_1)$, and no matter which strategy $g \in \S_2(B_2)$ \PT responds with, we have $\payoff(f, g) \geq c$. {\em \PO's value} is the maximal $c$ that he can guarantee, and \PT's value is defined dually. Note that there might be a gap between the two players' values. When \PO's value coincides with \PT's value, we say that the {\em value} exists in the game.

\subsection{Partial information bidding games}
A partial-information bidding game is $\G = \zug{V, E, \alpha, \gamma_1, \gamma_2}$, where $\zug{V,E}$ is a directed graph, $\alpha$ is an objective as we elaborate later, and the {\em budget distribution} $\gamma_i$ is a probability distribution from which \PLi's initial budget is drawn, for $i \in \set{1,2}$. The {\em support} of a probability distribution $\gamma: \Rat \rightarrow [0,1]$ is $\supp(\gamma) = \set{x \in \Rat: \gamma(x) > 0}$.  We restrict attention to finite-support probability distributions. For $i \in \set{1,2}$, the probability that \PLi's initial budget is $B_i \in \supp(\gamma_i)$ is $\gamma_i(B_i)$. 


\begin{definition}
{\bf One-sided partial information.}
We say that a game has one-sided partial information when $|\supp(\gamma_1)|=1$ and $|\supp(\gamma_2)| > 1$. We then call \PO the {\em partially-informed player} and \PT the {\em fully-informed player}. 
\end{definition}

We turn to define values in partial-information games. The intuition is similar to the full-information case only that each player selects a collection of strategies, one for each possible initial budget, and we take the expectation over the payoffs that each pair of strategies achieves. The $\delta$ in the following definition allows us to avoid corner cases due to ties in biddings and the $\epsilon$ is crucial to obtain the results on full-information mean-payoff bidding games.

\begin{definition}{\bf (Values in partial-information bidding games).}
\label{def:values}
Consider a partial-information bidding game $\G =\zug{V, E, \alpha, \beta, \gamma}$. Suppose $\supp(\beta) = \set{B_1, \ldots, B_n}$ and $\supp(\gamma) = \set{C_1, \ldots, C_n}$. We define \PO's value, denoted $\val^\downarrow(\G, \beta, \gamma)$, and \PT's value, denoted $\val^\uparrow(\G, \beta, \gamma)$, is defined symmetrically. We define that $\val^\downarrow(\G, \beta, \gamma) = c \in \Real$ if for every $\delta,\epsilon > 0$, 
\begin{itemize}
\item There is a collection $\big(f_B \in \S_1(B + \delta) \big)_{B \in \supp(\beta)}$ of \PO strategies, such that for every collection $\big(g_C \in \S_2(C) \big)_{C \in \supp(\gamma)}$ of \PT strategies, we have $\sum_{B,C} \beta(B) \cdot \gamma(C) \cdot \payoff(f_B, g_C) \geq c - \epsilon$.

\item For every collection $\big(f_B \in \S_1(B) \big)_{B \in \supp(\beta)}$ of \PO strategies, there is a collection $\big(g_C \in \S_2(C + \delta) \big)_{C \in \supp(\gamma)}$ of \PT strategies such that $\sum_{B,C} \beta(B) \cdot \gamma(C) \cdot \payoff(f_B, g_C) \leq c + \epsilon$.
\end{itemize}

Note that $\val^\downarrow(\G, \beta, \gamma) \leq \val^\uparrow(\G, \beta, \gamma)$ and when there is equality, we say that the value exists, and denote it by $\val(\G, \beta, \gamma)$. 
\end{definition}

The value in mean-payoff games is often called the {\em mean-payoff value}. In mean-payoff games we use $\MP^\downarrow, \MP^\uparrow$, and $\MP$ instead of $\val^\downarrow, \val^\uparrow$, and $\val$, respectively. When $\G$ is full-information and the budget ratio is $r$, we use $\MP(\G, r)$ instead of writing the two budgets.

\stam{
\begin{remark}{\bf (Revealing the initial budget).}
We note that the definition of strategies in partial- and full-information bidding games is the same; namely, a strategy is a function from histories to actions. Recall that histories only contain the previous bids performed. Thus, while in full-information games this suffices to compute the current budget of the opponent, in partial-information games a player only knows a distribution of possible current budgets of the opponent. In other words, when \PO plays according to $f$ and \PT according to $g_1$ and $g_2$, then as long as $g_1$ and $g_2$ bid the same, the bids made by $f$ are the same and so is the generated history. On the other hand, the opponent {\em reveals} her true initial budget when following a history $h$, we have $g_1(h) \neq g_2(h)$. 
\end{remark}
}

\section{Partial-Information Qualitative First-Price Bidding Games}
\label{sec:qual}
In this section, we focus on first-price bidding and show that the value exists in partial-information bidding games with qualitative objectives. The proof adapts results from the full-information setting, which we survey first. 


\begin{definition}
{\bf (Threshold ratios in full-information games).}
Consider a full-information first-price bidding game with a qualitative objective. Suppose that the sum of initial budgets is $1$ and that the game starts at $v$. The threshold ratio in $v$, denoted $\thresh(v)$, is a value $t$ such that for every $\epsilon > 0$:
\begin{itemize}
\item \PO wins when his ratio is greater than $\thresh(v)$; namely, when the initial budgets are $t+\epsilon$ and $1-t-\epsilon$. 
\item \PT wins when \PO's ratio is less than $\thresh(v)$; namely, when the initial budgets are $t-\epsilon$ and $1-t+\epsilon$. 
\end{itemize}
\end{definition}

Existence of threshold ratios for full-information reachability games was shown in \cite{LLPU96,LLPSU99} and later extended to full-information parity games in \cite{AHC19,AHI18}.
\begin{theorem}\cite{LLPU96,LLPSU99,AHC19,AHI18}
Threshold ratios exist in every vertex of a parity game.
\end{theorem}

The following theorem, 
extends these results to the partial-information setting. 


\begin{theorem}
\label{thm:partial-qual}
Consider a partial-information parity first-price bidding game $\G = \zug{V, E, \alpha, \beta, \gamma}$ and a vertex $v \in V$. Let $W = \set{\zug{B,C}: B \in \supp(\beta), \ C \in \supp(\gamma), \text{ and } \thresh(v) < \frac{B}{B+C}}$. Then, the value of $\G$ in $v$ is $\sum_{\zug{B,C} \in W} \beta(B) \cdot \gamma(C)$. 
\end{theorem}
\begin{proof}
Consider the following collection of strategies for \PO. For every $B \in \supp(\beta)$, let $C \in \supp(\gamma)$ be the maximal initial budget such that \PO wins with initial budgets $B$ and $C$ from $v$. That is, $C$ is the maximal element such that $\frac{B}{B+C} > \thresh(v)$. We fix \PO's strategy for initial budget $B$ to be a winning strategy $f$ against $C$. It is not hard to show that $f$ wins against any \PT strategy $g \in \S_2(C')$, for $C' < C$. 

To show that \Max cannot guarantee a higher payoff, we consider the dual collection of strategies for \Min: for every $C \in \supp(\gamma)$, \Min selects the maximal $B \in \supp(\beta)$ such that $\frac{B}{B+C} \leq \thresh(v)$, and plays according to a winning strategy for these budgets. Recall that we let \Min win bidding ties, thus she wins the game when $\frac{B}{B+C} = \thresh(v)$. Similar to the above, \Min wins for initial budgets $C$ and $B' < B$. 

To conclude, for each pair $\zug{B, C} \in \supp(\beta)\times \supp(\gamma)$, if $\frac{B}{B+C} > \thresh(v)$, \PO wins, and if $\frac{B}{B+C} \leq \thresh(v)$, \PT wins. Both players play irrespective of the opponent's strategy, hence the theorem follows. 
\end{proof}

\section{Partial-Information Mean-Payoff Bidding Games}
\label{sec:MP}
In this section we study mean-payoff bidding games. Throughout this section we focus on games played on {\em strongly-connected graphs}. We start by surveying results on full-information games. The most technically-challenging results concern one-sided partial-information poorman-bidding games. We first develop optimal strategies for the partially-informed player, and then show that the value does not necessary exist under pure strategies. 


\subsection{Full-information mean-payoff bidding games}\label{sec:MP-FI}
We show equivalences between bidding games and a class of stochastic games~\cite{Con92} called random-turn games, which are define formally as follows. 

\begin{definition}
{\bf (Random-turn games).} 
\label{def:RT}
Consider a strongly-connected mean-payoff bidding game $\G$. For $p \in [0,1]$, the random-turn game that corresponds to $\G$ w.r.t. $p$, denoted $\RT(\G, p)$, is a game in which instead of bidding, in each turn, we toss a (biased) coin to determine which player moves the token: \PO and \PT are respectively chosen with probability $p$ and $1-p$. Formally, $\RT(\G, p)$ is constructed as follows. Every vertex $v$ in $\G$, is replaced by three vertices $v_N, v_1$, and $v_2$. The vertex $v_N$ simulates the coin toss: it has an outgoing edge with probability $p$ to $v_1$ and an edge with probability $1-p$ to $v_2$. For $i \in \set{1,2}$, vertex $v_i$ simulates \PLi winning the coin toss: it is controlled by \PLi and has an outgoing edge to $u_N$, for every neighbor $u$ of $v$. The weights of $v_N, v_1$, and $v_2$ coincide with the weight of $v$. The mean-payoff value of $\RT(\G, p)$, denoted $\MP\big(\RT(\G, p)\big)$, is the optimal expected payoff that the two players can guarantee, and it is known to exist \cite{Put05}. Since $\G$ is strongly-connected, $\MP\big(\RT(\G, p)\big)$ does not depend on the initial vertex.
\end{definition}

For a full-information game $\G$ and a ratio $r \in (0,1)$, recall that $\MP(\G, r)$ denotes the optimal payoff that \Max can guarantee with initial ratio $r$. 
We state the equivalences between the two models.

\begin{theorem}
\label{thm:FI-MP}
Let $\G$ be a strongly-connected full-information mean-payoff bidding game. 
\begin{itemize}
\item {\bf First-price Richman bidding \cite{AHC19}.} The optimal payoff that \Max can guarantee with a pure strategy does not depend on the initial ratio: for every initial ratio $r$, we have $\MP(\G, r) = \MP\big(\RT(\G, 0.5)\big)$. 
\item {\bf First-price poorman bidding \cite{AHI18}.} The optimal payoff that \Max can guarantee with pure strategy and ratio $r$ coincides with the value of a random-turn game with bias $r$: for every initial ratio $r$, we have $\MP(\G, r) = \MP\big(\RT(\G, r)\big)$. 
\item {\bf All-pay poorman bidding \cite{AJZ21}.} The optimal payoff that \Max can guarantee with a pure strategy and ratio $r > 0.5$ coincides with the value of a random-turn game with bias $(2r-1)/r$: for every initial ratio $r > 0.5$, we have $\MP(\G, r) = \MP\big(\RT(\G, (2r-1)/r)\big)$. 
\end{itemize}
\end{theorem}

Since the optimal payoff under first-price Richman bidding depends only on the structure of the game and not on the initial ratios, the result easily generalizes to partial-information games. Consider two budget distributions $\beta$ and $\gamma$ for \Min and \Max, respectively. Indeed, when \Min's initial budget is $B \in \supp(\beta)$, playing optimally against any $C \in \supp(\gamma)$ results in the same payoff, and similarly for \Max. We thus conclude the following. 

\begin{theorem}
Consider a strongly-connected first-price Richman mean-payoff bidding game $\G$.
For any two budget distributions $\beta$ and $\gamma$ for the two players, we have $\MP^\downarrow(\G, \beta, \gamma) = \MP^\uparrow(\G, \beta, \gamma) = \MP\big(\RT(\G, 0.5)\big)$. 
\end{theorem}

\begin{remark}
{\bf (All-pay Richman bidding).}
It was shown in \cite{AJZ21} that in all-pay Richman bidding games, pure strategies are ``useless'': no matter what the initial ratio is, \Max cannot guarantee a positive payoff with a pure strategy. The study of mean-payoff all-pay Richman-bidding games is thus trivial in the partial-information setting as well. 
\end{remark}

\subsection{The value of the partially-informed player}
\label{sec:partially-informed}
We turn to study partial-information mean-payoff bidding games under poorman bidding, where we focus on one-sided partial information. We arbitrarily set \Max to be partially-informed and \Min to be fully-informed. 

\subsubsection{First-price bidding.}
Fix a strongly-connected mean-payoff game $\G$. Suppose that \Max's budget is $B$ and \Min's budget is chosen from a finite probability distribution $\gamma$ with $\supp(\gamma) = \set{C_1, \ldots, C_n}$ and $C_i < C_{i+1}$, for $1 \leq i < n$. We generalize the technique that is illustrated in Example~\ref{ex:partially-informed}. \Max carefully chooses increasing $x_1, \ldots, x_n$, where $x_n = B$. He maintains two ``accounts'': a spending account from which he bids and a savings account. Initially, the spending account has a budget of $x_1$ and the savings account, a budget of $B-x_1$. \Max plays ``optimistically''. He first plays in hope that \Min's budget is $C_1$ with a budget of $x_1$. If \Min does not spend $C_1$, the payoff is as in full-information games, namely at least $p_1 = \MP\big(\RT(\G, \frac{x_1}{x_1 + C_1})\big)$. Otherwise, \Min spends at least $C_1$ and \Max transfers budget from his savings account to his spending account so that the saving account has $B-x_2$ and the spending account has at least $x_2 - x_1$. Note that if \Min's initial budget was indeed $C_2$, at this point she is left with a budget of at most $C_2 - C_1$. If \Min does not spend $C_2 - C_1$, by following a full-information optimal strategy, \Max can guarantee a payoff of at least $p_2 = \MP\big(\RT(\G, \frac{x_2-x_1}{x_2-x_1 + C_2 - C_1})\big)$. The definition of $p_3,\ldots, p_n$ is similar. \Max chooses $x_1,\ldots, x_n$ so that $p_1 \geq \ldots \geq p_n$. Thus, when \Min's initial budget is $C_i$, she has an incentive to play so that \Max's spending account will reach $x_i$ and the payoff will be at least $p_i$. We call such a choice of $x_1,\ldots,x_n$ {\em admissible} and formally define it as follows. 

\begin{definition}\label{def:admissible}{\bf Admissible sequences.} 
Let $\G$ be a poorman mean-payoff bidding game. Let $B$ be a budget of \Max and $\gamma$ be a finite budget distribution of \Min
with $\supp(\gamma) =\{C_1,\dots,C_n\}$.
A sequence $(x_i)_{1 \leq i \leq n}$ of budgets is called
\emph{admissible} with respect to $B$ and $\gamma$ if 
$0 \leq x_1 \leq x_2 \leq \dots \leq x_n = B$ and
$p_1 \geq p_2 \geq \ldots \geq p_n$,~where
\begin{equation}\label{eq:xprelation}
    p_i = \MP\Big(\RT\Big(\G,\frac{x_i-x_{i-1}}{x_i-x_{i-1} + C_i - C_{i-1}}\Big)\Big)
\end{equation}
for each $1\leq i\leq n$, with $x_0=0$ and $C_0=0$.
We denote by $\Admi(B,\gamma)$ the set of all admissible sequences with respect to $B$ and $\gamma$.
\end{definition}

The main result of this section is stated in the following theorem. The upper bound is proven in Lemma~\ref{lem:partially-upper-bound} and the lower bound in Lemma~\ref{lem:partially-lower-bound}.

\begin{theorem}[Mean-payoff value of the partially-informed player]\label{thm:maxmeanpayoff}
Consider a strongly-connected first-price poorman mean-payoff bidding game $\G$. Let $B$ be the initial budget of \Max and $\gamma$ be a finite budget distribution of \Min with $\supp(\gamma) =\{C_1,\dots,C_n\}$. Then
\begin{equation}\label{eq:value}
\MP^\downarrow(\G, \beta, \gamma) = \max_{(x_i)_{1 \leq i \leq n} \in \Admi(B,\gamma)}\,   \mathsf{Val}(x_1,\dots,x_n),
\end{equation}
where  
\begin{equation}\label{eq:sum}
    \mathsf{Val}(x_1,\dots,x_n) = \sum_{i=1}^n   \gamma(C_i) \cdot \MP\Big( \RT\Big( \G, \frac{x_i-x_{i-1}}{x_i-x_{i-1} + C_i - C_{i-1}} \Big) \Big)
\end{equation}
with $x_0 = 0$ and $C_0 = 0$.
\end{theorem}

We point to some interesting properties of \Max's value:

\begin{remark}
\label{rem:partially}
Consider the bowtie game (Fig.~\ref{fig:bowtie}) and assume \Max's budget is fixed to $B=1$ and \Min's budget is drawn uniformly at random from $\set{C_1, C_2}$. 
\begin{itemize}
\item When $C_1 = 1$ and $C_2 = 2$, the maximum is obtained at $x= 0.5$, thus \Max's optimal expected payoff is $\frac{1}{3} = \frac{B}{B+C_2}$. We note that \Max has a very simple optimal strategy in this case: ``assume the worst'' on \Min's initial budget. That is, play according to an optimal strategy for initial budgets $B$ and $C_2$.
\item When $C_1 = 1$, and $C_2 = 5$, the maximum is obtained at $x=0$. This is the dual of the case above. \Max can ``assume the best'' on \Min's initial budget and play according to an optimal strategy for budgets $B$ and $C_1$. When \Min's budget is $C_1$, this strategy guarantees a payoff of $\frac{B}{B+C_1}$. But when \Min's budget is $C_2$, the strategy cannot guarantee a payoff above $0$. Thus, the strategy guarantees an expected payoff of $\frac{1}{4} = \frac{1}{2}\cdot\frac{B}{B+C_1}$. 
\item There are cases in which \Max's optimal strategy is not one of the trivial cases above. When $C_1 = 1$ and $C_2 = 3$, \Max's optimal payoff is $\frac{1}{8}(5-2\cdot\sqrt{2})\approx 0.271$, which is strictly larger than both $\frac{1}{4}=\frac{1}{2}\cdot\frac{B}{B+C_1}$ and $\frac{1}{4}=\frac{B}{B+C_2}$.\hfill$\triangle$
\end{itemize}
\end{remark}

\begin{definition}
We denote the right-hand-side of eq.~\eqref{eq:value} by $\mathsf{Val}$.
\end{definition}

\begin{lemma}[Upper bound]
\label{lem:partially-upper-bound}
Consider a strongly-connected first-price poorman mean-payoff bidding game $\G$. Let $B$ be the initial budget of \Max and $\gamma$ be a finite budget distribution of \Min with $\supp(\gamma) =\{C_1,\dots,C_k\}$. Then, for every $\epsilon>0$, \Max has a strategy that guarantees an expected mean-payoff of at least $\mathsf{Val}-\epsilon$.
\end{lemma}

\begin{proof}
Fix $\epsilon>0$.
For each $(x_i)_{1 \leq i \leq n} \in \Admi(B,\gamma)$,
we construct a \Max strategy $f_{x_1,\dots,x_n}$ that guarantees a payoff of at least
$\mathsf{Val}(x_1,\dots,x_n)-\epsilon$
as follows:

\begin{itemize}
    \item \Max uses portion $x_1$ of his budget to play an $\epsilon$-optimal strategy against \Min with budget $C_1$. This is continued as long as \Min spends at most $C_1$.
    \item For each $1\leq i\leq n-1$, once \Min's investments exceed $C_i$, \Max starts using portion $x_{i+1}-x_i$ of his budget and plays according to an $\epsilon$-optimal strategy against budget $C_{i+1}-C_i$ of $\Min$. This is continued as long as \Min's investments do not exceed $C_{i+1}$.
\end{itemize}


Lemma~\ref{lem:partially-upper-bound} follows from Claim~$1$ below, which generalizes the analysis of Example~\ref{ex:partially-informed} to any \Min finite budget distribution.
As in the example, it is crucial to select $x_i$ such that \Min has an incentive to ``reveal'' (when she can), that her budget is larger than $C_i$. Formally, recall that $p_i = (x_i-x_{i-1})/(x_i-x_{i-1} + C_i - C_{i-1})$. Intuitively, $p_i$ can be thought of as the payoff when \Max plays according to the strategy above and \Min's budget is $C_i$. Then, we require that $p_1 \geq \dots \geq p_n$. 


\noindent{\bf Claim 1.} For each $(x_i)_{1 \leq i \leq n} \in \Admi(B,\gamma)$, \Max ensures a payoff of at least $\mathsf{Val}(x_1,\dots,x_n)-\epsilon$ by playing according to the strategy $f_{x_1,\dots,x_n}$.

To prove Claim 1, fix a strategy $g$ of \Min and consider $\play(f_{x_1,\dots,x_n},g)$. Denote by $c$ the highest value of the budget lost by \Min during the course of the play, and let $1\leq i\leq n$ be such that $C_{i-1}< c\leq C_i$. Then, by the construction of $f_{x_1,\dots,x_n}$, the payoff of the play is at least $p_i-\epsilon$. Since we assumed that $p_1\geq p_2\geq \dots \geq p_n$ and since $\MP\big(\RT(\G,p)\big)$ is a monotonically decreasing function in $p$, it follows that the payoff of $\play(f_{x_1,\dots,x_n},g)$ is at least $p_i-\epsilon$ if $\Min$'s initial budget is $C_i$. Therefore, as the probability of \Min's initial budget being $C_i$ is $\gamma(C_i)$, we conclude that the expected payoff of $\play(f_{x_1,\dots,x_n},g)$ is at least $\mathsf{Val}(x_1,\dots,x_n)-\epsilon$. Since the strategy $g$ of $\Min$ was arbitrary, Claim~1 follows.
\end{proof}

Recall that \Max guarantees an expected payoff of $c \in \Real$ if, intuitively, he can reveal the strategy that he plays according to and no matter how \Min responds, the expected payoff is at least $c$. Thus, in order to show a lower bound on \Max's value, we show that no matter which strategy \Max chooses, \Min can respond in a way that guarantees a payoff of at most $\mathsf{Val} + \epsilon$. Formally we have the following.

\begin{lemma}[Lower bound]\label{lem:partially-lower-bound}
    Given $\epsilon>0$ and a strategy $f$ of \Max, there exist strategies $g_1\in \S(C_1),\dots,g_n\in \S(C_n)$ of \Min such that $\sum_{i=1}^n \gamma(C_i)\cdot \payoff(f,g_i) \leq \mathsf{Val} + \epsilon$. 
\end{lemma}

\begin{proof}
Let $\epsilon > 0$ and suppose that \Max plays according to a strategy $f$. As a response, for each $1 \leq i \leq n$, when \Min's initial budget is $C_i$, she selects an $\epsilon$-optimal response strategy $g_i \in \S_2(C_i)$ against $f$. We show that the choice of $g_1,\ldots, g_n$ satisfies the claim. 

Intuitively, we find an admissible sequence $x_1, \ldots, x_n$ and a corresponding ``wallet-based'' strategy $f_{x_1,\dots,x_n}$ as constructed in the proof of Lemma~\ref{lem:partially-upper-bound}, and show that $f_{x_1,\dots,x_n}$ achieves a payoff no worse than $f$ against $g_1,\ldots, g_n$. The proof follows since 
\[ \sum_{i=1}^n \gamma(C_i)\cdot \payoff(f,g_i) = \mathsf{Val}(x_1,\dots, x_n)+\epsilon \leq \mathsf{Val} + \epsilon. \]

To construct the admissible sequence $(x_i)_{1 \leq i \leq n} \in \Admi(B,\gamma)$, we set $x_n=B$ and define the remaining $x_i$'s as follows. Let $p_i = \payoff(f,g_i)$ for each $1\leq i\leq n$. Since $C_1 < \dots < C_n$, we have $p_1\geq \dots \geq p_n$. By Theorem~\ref{thm:FI-MP}, we have $\MP(\RT(\G, 0)) \leq p_n \leq \MP(\RT(\G, \frac{B}{B+C_n}))$. On the other hand, it is known~\cite{Cha12,Sol03} that the value $MP(\RT(\G, p))$ is a continuous function in $p$. Hence, there exists $B\cdot\frac{C_{n-1}}{C_n} \leq x \leq B$ such that $p_n = \MP(\RT(\G,\frac{B-x}{B-x+C_n-C_{n-1}}))$. We set $x_{n-1}$ to be the largest such~$x$.

We claim that, when the initial budget of \Min is $C_{n-1}$, \Max does not spend more than $x_{n-1}$ in the $\play(f,g_{n-1})$. Indeed, suppose towards contradiction that $\Max$ spends $x'>x_{n-1}$ while playing against $g_{n-1}$. Then, if the initial budget of \Min was $C_n$ and \Min used portion $C_{n-1}$ of her budget to play according to $g_{n-1}$, \Max would eventually be left with a budget of $B-x'<B-x_{n-1}$ and \Min would be left with at least $C_n-C_{n-1}$. Thus, \Min could play optimally for an initial budget of at least $C_n-C_{n-1}$ against a \Max budget smaller than $B-x'$, to ensure a payoff of at most $\MP(\RT(\G,\frac{B-x'}{B-x'+C_n-C_{n-1}})) \leq \MP(\RT(\G,\frac{B-x_{n-1}}{B-x_{n-1}+C_n-C_{n-1}})) = p_n$.
This would contradict either the optimality of $g_n$ in the case of strict inequality, or the maximality of $x_{n-1}$ in the case of equality. We thus conclude that \Max spends at most $x_{n-1}$ in $\play(f,g_{n-1})$.

Next, we define $x_{n-2}$. Note that the fact that \Max spends at most $x_{n-1}$ in $\play(f,g_{n-1})$ also implies that $\MP(\RT(\G, 0)) \leq p_{n-1}\leq \MP(\RT(\G,x_{n-1}/(x_{n-1}+C_{n-1}))$. Thus, as $MP(\RT(\G, p))$ is continuous in $p$, there exists $x_{n-1}\cdot C_{n-2}/C_{n-1} \leq x\leq x_{n-1}$ with $p_{n-1} = \MP(\RT(\G,\frac{x_{n-1} - x}{x_{n-1}-x+C_{n-1}-C_{n-2}}))$. Set $x_{n-2}$ to be the largest such $x$. Then, the same argument as above shows that \Max does not lose more than $x_{n-2}$ in the $\play(f,g_{n-2})$. We may then inductively repeat this procedure in order to define  $x_{n-3},\dots,x_1$. Note that this results in a sequence $0\leq x_1\leq x_2\leq \dots\leq x_n=B$ which by construction satisfies eq.~\eqref{eq:xprelation} for each $1\leq i\leq n$. Since we already showed that $p_1\geq \dots\geq p_n$, it follows that $(x_i)_{1 \leq i \leq n} \in \Admi(B,\gamma)$.
\end{proof}

\subsubsection{All-pay poorman bidding}
We extend the technique in the previous section to all-pay poorman bidding. 
In order to state our results formally, we need to redefine the notion of admissible sequences since the optimal payoff that \Max can guarantee under all-pay bidding differs from the payoff that he can guarantee under first-price bidding. Analogously to Def.~\ref{def:admissible} but now under all-pay bidding, we say that a sequence $(x_i)_{1 \leq i \leq n}$ of budgets is called \emph{admissible} with respect to a budget $B$ of \Max and a budget distribution $\gamma$ of \Min if  $0 \leq x_1 \leq x_2 \leq \dots \leq x_n = B$ and
$p_1 \geq p_2 \geq \ldots \geq p_n$,~where now
\begin{equation*}
    p_i = \MP\Big(\RT\Big(\G,\Big(1-\frac{C_i - C_{i-1}}{x_i-x_{i-1}}\Big)\cdot\mathbb{I}\Big(x_i-x_{i-1} > C_i - C_{i-1}\Big)\Big)\Big)
\end{equation*}
for each $1\leq i\leq n$, with $x_0=0$ and $C_0=0$. Here, $\mathbb{I}$ is an indicator function that evaluates to $1$ if the input logical formula is true, and to $0$ if it is false. We are now ready to state our result on all-pay poorman mean-payoff bidding games.

\begin{theorem}[Mean-payoff value of the partially-informed player]\label{thm:APmaxmeanpayoff}
Consider a strongly-connected all-pay poorman mean-payoff bidding game $\G$. Let $B$ be the initial budget of \Max and $\gamma$ be a finite budget distribution of \Min with $\supp(\gamma) =\{C_1,\dots,C_n\}$. Then
\begin{equation}\label{eq:APvalue}
\MP^\downarrow(\G, \beta, \gamma) = \max_{(x_i)_{1 \leq i \leq n} \in \Admi(B,\gamma)}\,   \mathsf{Val}(x_1,\dots,x_n),
\end{equation}
where 
$\mathsf{Val}(x_1,\dots,x_n) = \sum_{i=1}^n   \gamma(C_i) \cdot \MP\Big( \RT\Big( \G, \Big(1-\frac{C_i - C_{i-1}}{x_i-x_{i-1}}\Big)$$\cdot\mathbb{I}\Big(x_i-x_{i-1} > C_i - C_{i-1}\Big) \Big) \Big)$
\stam{
\begin{equation}\label{eq:sum}
    \mathsf{Val}(x_1,\dots,x_n) = \sum_{i=1}^n   \gamma(C_i) \cdot \MP\Big( \RT\Big( \G, \Big(1-\frac{C_i - C_{i-1}}{x_i-x_{i-1}}\Big)\cdot\mathbb{I}\Big(x_i-x_{i-1} > C_i - C_{i-1}\Big) \Big) \Big)
\end{equation}
}
with $x_0 = 0$ and $C_0 = 0$ and $\mathbb{I}$ an indicator function.
\end{theorem}

We introduce the following notation:
\begin{definition}
We denote the right-hand-side of eq.~\eqref{eq:APvalue} by $\mathsf{Val}$.
\end{definition}

The proof of the upper bound is similar to the proof for first-price poorman bidding and we include it for completeness.

\begin{lemma}[Upper bound]
\label{lem:partially-upper-bound-allpay}
Consider a strongly-connected all-pay poorman mean-payoff bidding game $\G$. Let $B$ be the initial budget of \Max and $\gamma$ be a finite budget distribution of \Min with $\supp(\gamma) =\{C_1,\dots,C_k\}$. Then, for every $\epsilon>0$, \Max has a strategy that guarantees an expected mean-payoff of at least $\mathsf{Val}-\epsilon$.
\end{lemma}
\begin{proof}
Fix $\epsilon>0$. For each $(x_i)_{1 \leq i \leq n} \in \Admi(B,\gamma)$, we construct a \Max strategy $f_{x_1,\dots,x_n}$ that guarantees a payoff of at least
$\mathsf{Val}(x_1,\dots,x_n)-\epsilon$
as follows:
\begin{itemize}
    \item \Max uses portion $x_1$ of his budget to play an $\epsilon$-optimal strategy against \Min with budget $C_1$. This is continued as long as \Min spends at most $C_1$.
    \item For each $1\leq i\leq n-1$, once \Min's investments exceed $C_i$, \Max starts using portion $x_{i+1}-x_i$ of his budget and plays according to an $\epsilon$-optimal strategy against budget $C_{i+1}-C_i$ of $\Min$. This is continued as long as \Min's investments do not exceed $C_{i+1}$.
\end{itemize}
Lemma~\ref{lem:partially-upper-bound-allpay} follows immediately from Claim~$1$ below. Recall that, for all-pay poorman mean-payoff bidding games, we defined $p_i = \MP(\RT(\G,(1-\frac{C_i - C_{i-1}}{x_i-x_{i-1}})\cdot\mathbb{I}(x_i-x_{i-1} > C_i - C_{i-1})))$.

\paragraph{Claim 1.} For each $(x_i)_{1 \leq i \leq n} \in \Admi(B,\gamma)$, \Max ensures a payoff of at least $\mathsf{Val}(x_1,\dots,x_n)-\epsilon$ by playing according to the strategy $f_{x_1,\dots,x_n}$.

To prove Claim 1, fix a strategy $g$ of \Min and consider $\play(f_{x_1,\dots,x_n},g)$. Denote by $c$ the highest value of the budget lost by \Min during the course of the play, and let $1\leq i\leq n$ be such that $C_{i-1}< c\leq C_i$. Then, by the construction of $f_{x_1,\dots,x_n}$, the payoff of the play is at least $p_i-\epsilon$. Since we assumed that $p_1\geq p_2\geq \dots \geq p_n$ and since $\MP\big(\RT(\G,p)\big)$ is a monotonically decreasing function in $p$, it follows that the payoff of $\play(f_{x_1,\dots,x_n},g)$ is at least $p_i-\epsilon$ if $\Min$'s initial budget is $C_i$. Therefore, as the probability of \Min's initial budget being $C_i$ is $\gamma(C_i)$, we conclude that the expected payoff of $\play(f_{x_1,\dots,x_n},g)$ is at least $\mathsf{Val}(x_1,\dots,x_n)-\epsilon$. Since the strategy $g$ of $\Min$ was arbitrary, Claim~1 follows.
\end{proof}

The proof of the lower bound is also similar to the proof for first-price poorman bidding but it requires care. In particular, if the initial budget of \Min is $C_i>B$, then \Min can guarantee an arbitrarily small payoff against any strategy of \Max according to Theorem~\ref{thm:FI-MP}. We need to take this into account when constructing an admissible sequence $x_1, \ldots, x_n$ and a corresponding ``wallet-based'' strategy $f_{x_1,\dots,x_n}$.

\begin{lemma}[Lower bound]\label{lem:partially-lower-bound-allpay}
    Given $\epsilon>0$ and a strategy $f$ of \Max, there exist strategies $g_1\in \S(C_1),\dots,g_n\in \S(C_n)$ of \Min such that $\sum_{i=1}^n \gamma(C_i)\cdot \payoff(f,g_i) \leq \mathsf{Val} + \epsilon$. 
\end{lemma}

\begin{proof}
Let $\epsilon > 0$ and suppose that \Max plays according to a strategy $f$. As a response, for each $1 \leq i \leq n$, when her initial budget is $C_i$, \Min selects a strategy $g_i \in \S_2(C_i)$ that is $\epsilon$-optimal against $f$. We show that the choice of $g_1,\ldots, g_n$ satisfies the claim. 

First, if $B<C_i$ for each $1\leq i\leq n$, then from eq.~\eqref{eq:value} and eq.~\eqref{eq:sum} we see that $\mathsf{Val}=0$. On the other hand, $\Max$ cannot guarantee any payoff better than $0$ against any possible budget of \Min, thus for each $i$ and for each strategy of $\Max$ there exists a response strategy of \Min that ensures payoff of at most $\epsilon$. Therefore, by our choice of $g_1,\dots,g_n$ we deduce that $\sum_{i=1}^n \gamma(C_i)\cdot \payoff(f,g_i)\leq \sum_{i=1}^n \gamma(C_i)\cdot \epsilon = \epsilon = \mathsf{Val} + \epsilon$, as desired.

Now, assume that there exists some $C_i<B$ and let $i^{\ast}$ be the largest such index. To prove that our choice of $g_1,\ldots, g_n$ satisfies the claim, we find an admissible sequence $x_1, \ldots, x_n$ and a corresponding ``wallet-based'' strategy $f_{x_1,\dots,x_n}$ as constructed in the proof of Lemma~\ref{lem:partially-upper-bound-allpay}, and show that $f_{x_1,\dots,x_n}$ achieves a payoff no worse than $f$ against $g_1,\ldots, g_n$. The proof follows since 
\[ \sum_{i=1}^n \gamma(C_i)\cdot \payoff(f,g_i) \leq \mathsf{Val}(x_1,\dots, x_n)+\epsilon \leq \mathsf{Val} + \epsilon. \]
\end{proof}

\subsection{The mean-payoff value of the fully-informed player under first-price poorman bidding}
\label{sec:fully}
In this section we identify the optimal expected payoff that the fully-informed player can guarantee in the bowtie game (Fig.~\ref{fig:bowtie}) under first-price bidding. Suppose that \Max's initial budget is $B$ and \Min's initial budget is drawn from a distribution $\gamma$. Consider the following collection of naive strategies for \Min: when her initial budget is $C \in \supp(\gamma)$, \Min plays according to an optimal full-information strategy for the ratio $\frac{B}{B+C}$. 

We find it surprising that this collection of strategies is optimal for \Min in the bowtie game. The technical challenge in this section is the lower bound. This result complements Thm.~\ref{thm:maxmeanpayoff}: we characterize both \Min and \Max's values in the bowtie game when the players are restricted to use pure strategies. We show, somewhat unexpectedly, that the two values do not necessarily coincide. 

In order to state the result formally, we need the following definition. Intuitively, the {\em potential} of $\zug{B, \gamma}$ is the optimal expected payoff when \Min plays according to the collection of naive strategies described above. 


\begin{definition}
{\bf (Potential).}
Given a budget $B \in \mathbb{R}$ of \Max and a budget distribution $\gamma$
with support $\supp(\gamma)=\{C_1,C_2,\dots,C_k\}$ of \Min,
we define $\Genpot{B}{\gamma} =\sum_{j=1}^k \gamma(C_j) \cdot \frac{B}{B+ C_j}$.
\end{definition}

\smallskip
The main result in this section is given in the following theorem, whose proof follows from Lemmas~\ref{lemma:up} and~\ref{lemma:low}.

\begin{theorem}[Mean-payoff value of the fully-informed player]\label{thm:fully-informed}
Consider the bowtie game $\lolli$.
Let $B$ be the initial budget of \Max and $\gamma$
be a finite budget distribution of \Min with $\supp(\gamma) =\{C_1,C_2,\dots,C_k\}$.
Then, 
\begin{equation}\label{eq:valueIMPERFECT}
    \MP^\uparrow(\G, B, \gamma) = \Genpot{B}{\gamma} = \sum_{j=1}^k   \gamma(C_j) \cdot \frac{B}{B + C_j}.
\end{equation}
\end{theorem}

Before proving the theorem, we note the following.
\begin{remark}
{\bf (Inexistence of a value).} 
Our result implies that the value in partial-information mean-payoff first-price poorman bidding games under pure strategies is not guaranteed to exist. Indeed, consider $\lolli$ with $B = 1$ and $\gamma$ that draws \Min's budget uniformly at random from $\set{1,2}$. By Thm.~\ref{thm:maxmeanpayoff}, one can verify that the optimal choice of $x$ is $1$, thus $\MP^\downarrow(\lolli, B, \gamma) = \frac{1}{3}$. On the other hand, by Thm.~\ref{thm:fully-informed}, we have $\MP^\uparrow(\lolli, B, \gamma) = \frac{5}{12}$. \hfill$\triangle$
\end{remark}

The upper bound is obtained when \Min reveals her true budget immediately and plays according to the strategies described above. The following lemma follows from results on full-information games (Thm.~\ref{thm:FI-MP}). 

\begin{lemma}[Upper bound]\label{lemma:up}
For every $\epsilon>0$, \Min has a collection
of strategies ensuring an expected payoff smaller than
$\Genpot{B}{\gamma} + \epsilon$.
\end{lemma}
\stam{
\begin{proof}
Let $\epsilon>0$.
\Min ensures the desired payoff by always using an optimal strategy
for her actual initial budget against the initial budget $B$ of Max.
Formally, for every $1 \leq j \leq k$,
when \Min starts with $C_j$ she has a strategy $g_j$ ensuring
the payoff $\frac{B}{B + C_j} - \epsilon$ (Thm.~\ref{thm:FI-MP}).
Then for every strategy $f$ of Max we get
\[
\sum_{j=1}^k \gamma(C_j) \cdot \payoff(f, g_j)
= \sum_{j=1}^k \gamma(C_j) \cdot \frac{B}{B + C_j} - \epsilon,
\]
as desired.
\end{proof}
}

We proceed to the more challenging lower bound and show that there are no \Min strategies that perform better than the naive strategy above.

\begin{lemma}[Lower bound]\label{lemma:low}
For every $\epsilon>0$ and
for every collection $(g_j \in \S_\Min(C_j))_{1 \leq j \leq k}$ of \Min strategies,
\Max has a strategy ensuring an expected payoff greater than $\Genpot{B}{\gamma} - \epsilon$.
\end{lemma}
\begin{proof}
Let $\epsilon > 0$, and let $(g_j \in \S_\Min(C_j))_{1 \leq j \leq k}$ be a collection of \Min strategies.
We construct a counter strategy $f$ of \Max
ensuring an expected payoff greater than $\Genpot{B}{\gamma} - \epsilon$.
The proof is by induction over the size $k$ of the support of $\gamma$.
Obviously, if $k=1$, \Max has perfect information and can follow a full-information optimal strategy to guarantee a payoff of 
$\Genpot{B}{\gamma} = \frac{B}{B+C_1}$ (Thm.~\ref{thm:FI-MP}).
So suppose that $k>1$, and that the statement holds for every budget distribution of \Min
with a support strictly smaller than $k$.

\Max carefully chooses a small part $x \leq B$ of his budget and a part $y \leq C_1$ of \Min's budget. He plays according to a full-information strategy $f$ for initial budgets $x$ and $y$. 
This can result in three possible outcomes: {\bf (O$_1$)} \Min never uses more than $y$: the payoff is $\frac{x}{x+y}$  as in full-information games; {\bf (O$_2$)} \Min reveals her true initial budget, thus \Max can distinguish between the case that \Min's budget is $C_i$ and $C_j$, and by the induction hypothesis he can ensure an expected payoff of $\Genpot{B-x}{\gamma}$ using his remaining budget;
{\bf (O$_3$)} \Min does not reveal her true initial budget and spends more than $y$:  \Max's leftover budget is greater than $B-x$ and, for $1 \leq j \leq k$, when \Min's budget is $C_j$, she has $C_j-y$, and \Max re-starts the loop by selecting a new $x$. 

\stam{
\begin{description}
    \item[O$_1$] \Min never uses more than $y$. Thus, the outcome is as in the full-information setting, and the payoff is $\frac{x}{x+y}$;
    \item[O$_2$] \Min reveals her true initial budget; namely, there are $g_i$ and $g_j$ that prescribe different actions. 
    Then, \Max can distinguish between the case that \Min's budget is $C_i$ and $C_j$, and by the induction hypothesis he can play with his remaining budget to ensure the expected payoff $\Genpot{B-x}{\gamma}$;
    \item[O$_3$] \Min does not reveal her true initial budget and spends more than $y$.
    Thus, \Max's leftover budget is greater than $B-x$ and, for $1 \leq j \leq k$, when \Min's budget is $C_j$, she has $C_j-y$. \Max re-starts the loop and selects a new $x$. 
\end{description}
}

We show that \Max can choose $x$ and $y$ in a way that guarantees that the payoffs obtained
in the first two outcomes are greater than the desired payoff $\Genpot{B}{\gamma} - \epsilon$. 
Also, outcome O$_3$ can occur only finitely many times and the potential there does not decrease. Thus, O$_1$ or O$_2$ occur, ensuring a payoff of at least $\Genpot{B}{\gamma} -\epsilon$. 

Formally, we describe a sequence $(\pi_i,B_i,\gamma_i)_{0 \leq i \leq m}$ of configurations
comprising of a history $\pi_i$ consistent with every strategy $(g_j)_{1 \leq j \leq k}$,
the budget $B_i$ of \Max after $\pi_i$,
and the budget distribution $\gamma_i$ of \Min with $\supp(\gamma_i) =  \{C_1^i,C_2^i, \ldots, C_k^i\}$ following $\pi_i$.
Tuple~$i$ represents the budget and budget distribution of the players following $i-1$ choices of outcome O$_3$. 
Let $\lambda = 1 - \frac{\epsilon}{2}$ and $\rho = \frac{1}{\Genpot{B}{\gamma}} -1$.
We start with $(\pi_0,B_0,\gamma_0) = (v,B,\gamma)$ with $v$ an initial vertex,
and we show recursively how Max can update this tuple while ensuring
that the following four properties are satisfied:    
\begin{itemize}[noitemsep,topsep=0pt]
    \item[P$_1$:]
    The history $\pi_i$ is consistent with every $(g_j)_{1 \leq j \leq k}$;
    \item[P$_2$:]
    Max spends his budget sufficiently slowly:
    $B_i \geq \lambda^i B$;
    \item[P$_3$:]
    Min spends her budget sufficiently fast:
    $C_{j}^{i} \leq C_j - \rho \cdot (1-\lambda^i) B$ for every $1 \leq j \leq k$;
    \item[P$_4$:]
    The potential never decreases: $\Genpot{B_i}{\gamma_i} \geq \Genpot{B}{\gamma}$.
\end{itemize}

Note that for the initial tuple $(\pi_0,B_0,\gamma_0) = (v,B,\gamma)$,
these are trivially satisfied.
Moreover, Property P$_3$ implies an upper bound on $i$, that is, outcome O$_3$ can happen only finitely many times: $\lim_{i \rightarrow \infty} C_{1}^{i} \leq \lim_{i \rightarrow \infty}  C_1 -\rho \cdot (1-\lambda^i) B
= C_1 - \rho \cdot B = B+ C_1 - \frac{B}{\Genpot{B}{\gamma}}
= \frac{1}{\frac{1}{B+C_1}} - \frac{1}{
\sum_{j=1}^k \gamma(C_j) \cdot \frac{1}{B + C_j}}$
\stam{
\begin{align*}
\lim_{i \rightarrow \infty} C_{1}^{i}
&\leq \lim_{i \rightarrow \infty}  C_1 -\rho \cdot (1-\lambda^i) B
= C_1 - \rho \cdot B
= B+ C_1 - \frac{B}{\Genpot{B}{\gamma}}\\
&= \frac{1}{\frac{1}{B+C_1}} - \frac{1}{
\sum_{j=1}^k \gamma(C_j) \cdot \frac{1}{B + C_j}}
\end{align*}

\begin{equation*}
\begin{split}
&\lim_{i \rightarrow \infty} C_{1}^{i} \leq \lim_{i \rightarrow \infty}  C_1 -\rho \cdot (1-\lambda^i) B
= C_1 - \rho \cdot B \\
&= B+ C_1 - \frac{B}{\Genpot{B}{\gamma}}
= \frac{1}{\frac{1}{B+C_1}} - \frac{1}{
\sum_{j=1}^k \gamma(C_j) \cdot \frac{1}{B + C_j}}
\end{split}
\end{equation*}
}
which is negative since $C_1<C_2<\ldots<C_k$, yet a negative $C^i_1$ means that \Min illegally bids higher than her available budget.

We now define the choices $x_i$ and $y_i$ for each $i \in \Nat$, and show that they satisfy the properties described above. 
Let $x_i = \frac{\epsilon}{2} \cdot \lambda^{i} B$
and $y_i = \rho \cdot x_i$. 
For initial budgets $x_i$ and $y_i$, let $f_i$ be a full-information \Max strategy whose payoff is greater than $\frac{x_i}{x_i + y_i} - \epsilon$. \Max follows $f_i$ as long as \Min spends at most $y_i$. Let $(\psi_j)_{1 \leq j \leq k}$ be plays such that for each $1 \leq j \leq k$:
\begin{itemize}[noitemsep,topsep=0pt]
    \item the play $\pi_i\psi_j$ is consistent with the strategy $g_j$;
    \item Max plays according to $f_i$ along $\psi_j$;
    \item $\psi_j$ stops when Min uses more than $y_i$,
    and is infinite if she never does.
\end{itemize}


We consider three possible cases,
depending on whether the paths $\psi_j$ are finite or infinite,
and whether they are distinct or identical.
If they are all infinite, or there are at least two distinct ones,
we show that Max immediately has a way to obtain the desired payoff.
If they are all identical and finite, we show that,
while Max cannot immediately get the desired payoff, 
he can go to the next step by setting $\phi_{i+1} = \phi_i \psi_1$, 
and restarting. 

\noindent
{\bf 1. The play $\psi_j$ is infinite for every $1 \leq j \leq k$.}

\noindent
This situation happens if Min does not spend more than $y_i$.
Since \Max follows the strategy $f_i$ along each $\psi_j$,
the resulting payoff is greater than $\frac{x_i}{x_i + y_i} - \epsilon$. Moreover, the definition of $y_i$ implies:
\[
\frac{x_i}{x_i+y_i} = \frac{x_i}{x_i+\rho \cdot x_i} = \frac{x_i}{x_i+(\frac{1}{\Genpot{B}{\gamma}} - 1)x_i} =
\Genpot{B}{\gamma}.
\]

\noindent
{\bf 2. The plays $(\psi_j)_{1 \leq j \leq k}$ are not all identical.}

\noindent
Let $P_1,P_2,\ldots,P_m$ be the partition of $\{1,2,\ldots,k\}$
such that for every pair $1 \leq j,j' \leq k$,
the plays $\psi_j$ and $\psi_{j'}$ are equal if and only if
$j$ and $j'$ belong to the same $P_\ell$.
Remark that $m \geq 2$ since by supposition the plays $(\psi_j)_{1 \leq j \leq k}$ are not all identical.
We show that \Max can follow some $\psi_j$
until he identifies precisely which $P_\ell$ corresponds to the initial budget of \Min,
which allows us to apply the induction hypothesis,
and to show that \Max can guarantee the desired payoff.

For each $1 \leq \ell \leq m$, the plays $(\psi_j)_{j \in P_\ell}$ are equal by definition,
and we denote this play by~$\chi_\ell$.
We start by trimming the infinite plays into finite plays that still allow \Max to determine the adequate~$P_\ell$:
for every $1 \leq \ell \leq m$,
let $\chi_\ell'$ be a finite prefix of $\chi_\ell$ that is only consistent
with the strategies $g_j$ of \Min satisfying $j \in P_\ell$
(note that if the play $\chi_\ell$ is already finite,
we can set $\chi_\ell' = \chi_\ell$).
Remark  that the play $\chi_\ell'$ occurs with probability exactly $\sum_{j \in P_{\ell}} \gamma_i(C_j^i)$,
which we denote by~$\gamma(P_\ell)$.
After the play~$\phi_i \chi_\ell'$,
the remaining budget of Max is bigger that $B_i - x_i$.
Moreover, since this play is only consistent with the strategies $g_{j}$ of \Min satisfying $j \in P_\ell$,
Max knows that the current distribution of budgets of \Min
is the function $\gamma_{i.\ell}$ defined~by
\[
\gamma_{i.\ell}(C^i_j-y) = \frac{\gamma_i(C^i_j)}{\gamma(P_\ell)},
\]
where $j \in P_{\ell}$ and $y$ denotes the budget spent by \Min along $\chi_\ell'$.

Since $|P_\ell|<k$, the induction hypothesis implies that from this point
\Max can guarantee an expected payoff greater than~$\Genpot{B_i - x_i}{\gamma_{i.\ell}}-\frac{\epsilon}{2}$.
This holds for every $1 \leq \ell \leq m$, therefore
Max can globally guarantee an expected payoff greater than
\begin{equation*}
\begin{split}
    &\sum_{\ell = 1}^m \gamma(P_\ell) \cdot \Genpot{B_i-x_i}{\gamma_{i.\ell}} -\frac{\epsilon}{2} \\
    &= \sum_{\ell = 1}^m \gamma(P_\ell) \cdot \sum_{j \in P_{\ell}} \gamma_{i.\ell}(C_j^i-y) \frac{B_i-x_i}{B_i-x_i + C_j^i-y} -\frac{\epsilon}{2}\\
&\geq
\sum_{\ell = 1}^m \gamma(P_\ell) \cdot \sum_{j \in P_{\ell}} \frac{\gamma_{i}(C_j^i)}{\gamma(P_\ell)} \frac{B_i-x_i}{B_i-x_i + C_j^i}-\frac{\epsilon}{2} \\
&= \sum_{\ell = 1}^m \sum_{j \in P_{\ell}} \gamma_{i}(B_j^i) \frac{B_i-x_i}{B_i-x_i + C_j^i}-\frac{\epsilon}{2}\\
&= \Genpot{B_i-x_i}{\gamma_i}-\frac{\epsilon}{2}
\end{split}
\end{equation*}

To conclude, we show that $\Genpot{B_i-x_i}{\gamma_i} \geq \Genpot{B}{\gamma} - \frac{\epsilon}{2}$.
For all $1 \leq j \leq k$, the definition of $x_i$ and Property P$_2$ imply
\begin{align*}
    \frac{B_i}{B_i+C_j^i} - \frac{B_i-x_i}{B_i-x_i+C_j^i} =
    \frac{C_j^ix_i}{(B_i+C_j^i)(B_i+C_j^i-x_i)} \leq \frac{x_i}{B_i} \leq \frac{\epsilon}{2}.
\end{align*}
Therefore $\Genpot{B_i-x_i}{\gamma_i} \geq \Genpot{B_i}{\gamma_i} - \frac{\epsilon}{2}$,
which translates to $\Genpot{B_i-x_i}{\gamma_i} \geq \Genpot{B}{\gamma} - \frac{\epsilon}{2}$ by Property P$_4$.\\

\noindent
{\bf 3. The plays $(\psi_j)_{1 \leq j \leq k}$ are identical and finite.}

\noindent
If the $\psi_j$ are all equal to a finite play $\psi$,
then we define $\pi_{i+1}$ as the concatenation of $\pi_i$ and $\psi$.
The budget $B_{i+1}$ is obtained by subtracting from $B_i$ the budget spent by \Max along $\psi$.
Moreover, for every $1 \leq j \leq k$,
the distribution $\gamma_{i+1}$ maps the budget $C_j^{i+1}$ obtained by 
subtracting from $C_j^{i}$ the budget spent by \Min along $\psi$
to the probability $\gamma_{i}(C_j^{i}) = \gamma(C_j) \in [0,1]$.

We show that the configuration $(\pi_{i+1},B_{i+1}, \gamma_{i+1})$ satisfies properties P$_1$-P$_4$.
First, Property P$_1$ holds as $\phi_i \psi$ is consistent with every $g_j$.
Second, since Max follows the strategy $f_i$ along $\psi$, he does not spend more than $x_i$.
Therefore, since his budget $B_i$ after $\phi_i$ satisfies Property P$_2$,
so does his budget $B_{i+1}$ after $\phi_i \psi$:
\[
B_{i+1} \geq B_i - x_i \geq
\lambda^iB - \frac{\epsilon}{2} \lambda^iB =
(1-\frac{\epsilon}{2})\lambda^i B =
\lambda^{i+1}B.
\]
Moreover, since Min needs to use more than $y_i$ in order for $\psi$ to stop,
we can also conclude P$_3$:
\begin{align*}  
C_j^{i+1} &\leq C_j^{i} - y_i \leq
C_j - \rho (1-\lambda^i)B - \rho \frac{\epsilon}{2} \lambda^iB
= C_j - \rho (1-\lambda^{i+1})B.
\end{align*}
Finally, we obtain Property P$_4$ as a consequence of Properties P$_2$ and P$_3$.
Let $x$ denote the overapproximation $(1-\lambda^{i+1})B$
of the budget spent by Max since the start of the game,
and let $y$ denote the underapproximation $\rho \cdot (1-\lambda^{i+1})B$
of the budget spent by Min since the start of the game.
Then
\begin{equation*}
\begin{split}
    &\Genpot{B_{i+1}}{\gamma_{i+1}}
 =  \sum_{j=1}^k \gamma(C_j) \frac{B_{i+1}}{B_{i+1} + C_j^{i+1}}
 \geq \ \sum_{j=1}^k \gamma(C_j)
\frac{B - x}{B -x + C_j - y} \\
&= \sum_{j=1}^k \gamma(C_j)
\frac{B}{B+C_j} \cdot \frac{B - x}{B - \frac{B}{B+C_j} (x+y)}
 =  \sum_{j=1}^k \gamma(C_j) f\Big(\frac{B}{B+C_j}\Big),
\end{split}
\end{equation*}
where $f$ is the function mapping $\lambda \in \mathbb{R}$ to $\lambda \cdot \frac{B-x}{B-\lambda (x+y)}$.
As $f$ is convex, we may apply Jensen's inequality,
and use the fact that $\Genpot{B}{\gamma} \cdot (x + y) = x$ to conclude that
\stam{
\begin{align}\label{equ:2}
\Genpot{B_{i+1}}{\gamma_{i+1}}
&\geq \ f\Big(\sum_{j=1}^k \gamma(C^j) \cdot \frac{B}{B+C^j}\Big)
 \ = \ f(\Genpot{B}{\gamma})
 \ = \ \frac{\Genpot{B}{\gamma} \cdot(B-x)}{B-\Genpot{B}{\gamma} \cdot (x+y)}\nonumber\\
&= \ \Genpot{B}{\gamma} \nonumber.\qed 
\end{align}
}
\begin{equation*}
\begin{split}
    \Genpot{B_{i+1}}{\gamma_{i+1}} &\geq  f\Big(\sum_{j=1}^k \gamma(C^j) \cdot \frac{B}{B+C^j}\Big)
  =  f(\Genpot{B}{\gamma}) \\
  &=  \frac{\Genpot{B}{\gamma} \cdot(B-x)}{B-\Genpot{B}{\gamma} \cdot (x+y)} =  \Genpot{B}{\gamma}. 
\end{split}
\end{equation*}
\end{proof}

\section{Discussion and Future Work}
\label{sec:disc}
We initiate the study of partial-information bidding games, and specifically bidding games with partially-observed budgets.
Our most technically challenging results are for one-sided partial-information mean-payoff poorman-bidding games. We show a complete picture in strongly-connected games for the partially-informed player, which is the more important case in practice. By identifying the value for the fully-informed player in the bowtie game, we show that the value in mean-payoff bidding games does not necessarily exist when restricting to pure strategies. 

We discuss open problems in this model. First, we focus on games played on strongly-connected graphs. Reasoning about such games is the crux of the solution to general full-information bidding games. We thus expect that our results will be key in the solution of partial-information bidding games on general graphs. This extension, however, is not straightforward as in the full-information setting, and we leave it as an open question.
Second, we identify the value of the fully-informed player in the bowtie game $\lolli$. Reasoning about $\lolli$ was the crux of the solution to general strongly-connected full-information bidding games. In fact, the same technique was used to lift a solution for $\lolli$ to general strongly-connected games under all the previously-studied bidding mechanisms.
In partial-information games, however, this technique breaks the intricate analysis in the proof of Thm.~\ref{thm:fully-informed}. Again, we expect a solution to the bowtie game to be a key ingredient in the solution to general strongly-connected games, and we leave the problem open. 
Finally, we showed that the value does not necessarily exist under pure strategies. We leave open the problem of developing optimal mixed strategies for the players. 



This work is part of a research that combines formal methods and AI including multi-agent graph games~\cite{AHK02}, 
logics to reason about strategies~\cite{CHP10,MMPV14} and in particular, their application in auctions~\cite{MM+22},
enhancing {\em network-formation games} with concepts from formal methods (e.g.,~\cite{AKT16}), and many more.

\small
\bibliographystyle{plain}
\bibliography{../../ga}

\begin{thebibliography}{10}

\bibitem{AAH21}
M.~Aghajohari, G.~Avni, and T.~A. Henzinger.
\newblock Determinacy in discrete-bidding infinite-duration games.
\newblock {\em Log. Methods Comput. Sci.}, 17(1), 2021.

\bibitem{AHK02}
R.~Alur, T.~A. Henzinger, and O.~Kupferman.
\newblock Alternating-time temporal logic.
\newblock {\em J. {ACM}}, 49(5):672--713, 2002.

\bibitem{AMS95}
Robert~J Aumann, Michael Maschler, and Richard~E Stearns.
\newblock {\em Repeated games with incomplete information}.
\newblock MIT press, 1995.

\bibitem{AH20}
G.~Avni and T.~A. Henzinger.
\newblock A survey of bidding games on graphs.
\newblock In {\em Proc. 31st CONCUR}, volume 171 of {\em LIPIcs}, pages
  2:1--2:21. Schloss Dagstuhl - Leibniz-Zentrum f{\"{u}}r Informatik, 2020.

\bibitem{AHC19}
G.~Avni, T.~A. Henzinger, and V.~Chonev.
\newblock Infinite-duration bidding games.
\newblock {\em J. ACM}, 66(4):31:1--31:29, 2019.

\bibitem{AHI18}
G.~Avni, T.~A. Henzinger, and R.~Ibsen-Jensen.
\newblock Infinite-duration poorman-bidding games.
\newblock In {\em Proc. 14th WINE}, volume 11316 of {\em LNCS}, pages 21--36.
  Springer, 2018.

\bibitem{AHZ21}
G.~Avni, T.~A. Henzinger, and D.~Zikelic.
\newblock Bidding mechanisms in graph games.
\newblock {\em J. Comput. Syst. Sci.}, 119:133--144, 2021.

\bibitem{AIT20}
G.~Avni, R.~Ibsen{-}Jensen, and J.~Tkadlec.
\newblock All-pay bidding games on graphs.
\newblock In {\em Proc. 34th AAAI}, pages 1798--1805. {AAAI} Press, 2020.

\bibitem{AJZ21}
G.~Avni, I.~Jecker, and {\DJ}.~\v{Z}ikeli\'c.
\newblock Infinite-duration all-pay bidding games.
\newblock In {\em Proc. 32nd SODA}, pages 617--636, 2021.

\bibitem{AKT16}
G.~Avni, O.~Kupferman, and T.~Tamir.
\newblock Network-formation games with regular objectives.
\newblock {\em Inf. Comput.}, 251:165--178, 2016.

\bibitem{Bor21}
E.~Borel.
\newblock La th\'eorie du jeu les \'equations int\'egrales \'a noyau
  sym\'etrique.
\newblock {\em Comptes Rendus de l'Acad\'emie}, 173(1304--1308):58, 1921.

\bibitem{Cha12}
K.~Chatterjee.
\newblock Robustness of structurally equivalent concurrent parity games.
\newblock In {\em Proc. 15th FoSSaCS}, pages 270--285, 2012.

\bibitem{CHP10}
K.~Chatterjee, T.~A. Henzinger, and N.~Piterman.
\newblock Strategy logic.
\newblock {\em Inf. Comput.}, 208(6):677--693, 2010.

\bibitem{Con92}
A.~Condon.
\newblock The complexity of stochastic games.
\newblock {\em Inf. Comput.}, 96(2):203--224, 1992.

\bibitem{DP10}
M.~Develin and S.~Payne.
\newblock Discrete bidding games.
\newblock {\em The Electronic Journal of Combinatorics}, 17(1):R85, 2010.

\bibitem{LLPSU99}
A.~J. Lazarus, D.~E. Loeb, J.~G. Propp, W.~R. Stromquist, and D.~H. Ullman.
\newblock Combinatorial games under auction play.
\newblock {\em Games and Economic Behavior}, 27(2):229--264, 1999.

\bibitem{LLPU96}
A.~J. Lazarus, D.~E. Loeb, J.~G. Propp, and D.~Ullman.
\newblock Richman games.
\newblock {\em Games of No Chance}, 29:439--449, 1996.

\bibitem{MKT18}
R.~Meir, G.~Kalai, and M.~Tennenholtz.
\newblock Bidding games and efficient allocations.
\newblock {\em Games and Economic Behavior}, 112:166--193, 2018.

\bibitem{MM+22}
M.~Mittelmann, B.~Maubert, A.~Murano, and L.~Perrussel.
\newblock Automated synthesis of mechanisms.
\newblock In {\em Proc. 31st IJCAI}, pages 426--432. ijcai.org, 2022.

\bibitem{MMPV14}
F.~Mogavero, A.~Murano, G.~Perelli, and M.~Y. Vardi.
\newblock Reasoning about strategies: On the model-checking problem.
\newblock {\em {ACM} Trans. Comput. Log.}, 15(4):34:1--34:47, 2014.

\bibitem{PSSW09}
Y.~Peres, O.~Schramm, S.~Sheffield, and D.~B. Wilson.
\newblock Tug-of-war and the infinity laplacian.
\newblock {\em J. Amer. Math. Soc.}, 22:167--210, 2009.

\bibitem{PR89}
A.~Pnueli and R.~Rosner.
\newblock On the synthesis of a reactive module.
\newblock In {\em Proc. 16th POPL}, pages 179--190, 1989.

\bibitem{Put05}
M.~L. Puterman.
\newblock {\em Markov Decision Processes: Discrete Stochastic Dynamic
  Programming}.
\newblock John Wiley \& Sons, Inc., New York, NY, USA, 2005.

\bibitem{RC+07}
J.{-}F. Raskin, K.~Chatterjee, L.~Doyen, and T.~A. Henzinger.
\newblock Algorithms for omega-regular games with imperfect information.
\newblock {\em Log. Methods Comput. Sci.}, 3(3), 2007.

\bibitem{Reif84}
J.~H. Reif.
\newblock The complexity of two-player games of incomplete information.
\newblock {\em J. Comput. Syst. Sci.}, 29(2):274--301, 1984.

\bibitem{Sol03}
E.~Solan.
\newblock Continuity of the value of competitive markov decision processes.
\newblock {\em Journal of Theoretical Probability}, 16:831--845, 2003.

\bibitem{DDR06}
M.~De Wulf, L.~Doyen, and J.{-}F. Raskin.
\newblock A lattice theory for solving games of imperfect information.
\newblock In {\em Proc. 9th HSCC}, volume 3927 of {\em LNCS}, pages 153--168.
  Springer, 2006.

\end{thebibliography}

\end{document}